%% file: displacement.tex
\documentclass[12pt,a4paper]{amsart}
\usepackage[dvips,colorlinks=true,linkcolor=blue]{hyperref}
\usepackage{subfigure,caption}
\usepackage{amssymb}
\usepackage{graphicx}
\usepackage[rflt]{floatflt}
\numberwithin{equation}{section}
\newcommand{\D}{\displaystyle}

\newcommand{\R}{{\mathbb R}}
\newcommand{\T}{{\mathbb T}}
\newcommand{\Z}{{\mathbb Z}}

\newcommand{\vers}{\operatornamewithlimits{\to}}
\newcommand{\er}{{\mathbb R}}
\newcommand{\en}{{\mathbb N}}
\newcommand{\zed}{{\mathbb Z}}
\theoremstyle{plain}
\newtheorem{Th}{Theorem}[section]
\newtheorem{Le}{Lemma}[section]
\newtheorem{Pro}{Proposition}[section]
\newtheorem{Cor}{Corollary}[section]
\theoremstyle{definition}
\newtheorem{Rem}{Remark}[section]

\title{Localization for the random displacement model at weak disorder}
\author{Fatma Ghribi} \author{Fr{\'e}d{\'e}ric Klopp}
\address[Fatma Ghribi]{D{\'e}partement de Math{\'e}matiques,
  Facult{\'e} des Sciences de Monastir \\
  Avenue de l'Environnement 5019 Monastir, Tunisie}
\email{\href{mailto:Fatma.Ghribi@fsm.rnu.tn}{Fatma.Ghribi@fsm.rnu.tn}}
\address[Fr{\'e}d{\'e}ric Klopp]{LAGA, Institut Galil{\'e}e, U.R.A 7539 C.N.R.S,
  Universit{\'e} de Paris-Nord, Avenue J.-B.  Cl{\'e}ment, F-93430
  Villetaneuse, France\\ et \\ Institut Universitaire de France}
\email{\href{mailto:klopp@math.univ-paris13.fr}{klopp@math.univ-paris13.fr}}
\keywords{}
\subjclass{}
\thanks{The work of F.K. was supported by the grant
  ANR-08-BLAN-0261-01.}

\begin{document}

\begin{abstract}
  This paper is devoted to the study of the random displacement model
  on $\R^d$.  We prove that, in the weak displacement regime, Anderson
  and dynamical localization hold near the bottom of the spectrum
  under a generic assumption on the single site potential and a fairly
  general assumption on the support of the possible
  displacements. This result follows from the proof of the existence
  of Lifshitz tail and of a Wegner estimate for the model under
  scrutiny.
  \vskip.5cm
  \par\noindent \textsc{R{\'e}sum{\'e}.}
  Cet article est consacr{\'e} {\`a} l'{\'e}tude d'un mod{\`e}le de petits
  d{\'e}placements al{\'e}atoires. Sous une hypoth{\`e}se g{\'e}n{\'e}rique sur le
  potentiel de simple site et des hypoth{\`e}ses assez g{\'e}n{\'e}rales sur les
  d{\'e}placements autoris{\'e}s, on d{\'e}montre que le bas du spectre est
  exponentiellement et dynamiquement localis{\'e} dans la limite des
  petits d{\'e}placements. La preuve repose sur la preuve d'une estim{\'e}e de
  Lifshitz et d'une estim{\'e}e de Wegner pour le mod{\`e}le {\'e}tudi{\'e}.
\end{abstract}
\setcounter{section}{-1}
\maketitle
\section{Introduction}
\label{sec:introduction}
We consider the following random displacement model
\begin{equation}
  \label{eq:1}
  H_{\lambda,\omega}=-\Delta + p + q_{\lambda,\omega}\text{ where }
  q_{\lambda,\omega}(x)=\sum_{\gamma \in\Z^d} q(x-\gamma- \lambda
  \omega_{\gamma}), 
\end{equation}
acting on $L^2({\er}^d)$. We assume the following:
\begin{description}
\item[(H.0.0)] The potential $p$ is a real valued, $\Z^d$-periodic
  function.
\item[(H.0.1)] The single site potential $q$ is a twice continuously
  dif\-ferentiable, real valued function and compactly supported.
\item[(H.0.2)] $\omega:=(\omega_{\gamma})_{\gamma \in {\zed}^d}$ is a
  collection of non trivial, independent, identically distributed,
  bounded random variables; let $K\subset\R^d$ be the support of their
  common distribution.
\item[(H.0.3)] $\lambda$ is a small positive coupling constant.
\end{description}
Under these assumptions, $H_{\lambda,\omega}$ is ergodic and, for all
$\omega$, $H_{\lambda,\omega}$ is self-adjoint on the standard Sobolev
space $\mathcal{H}^2({\er}^d)$. The theory of ergodic operators
teaches us that the spectrum of $H_{\lambda,\omega}$ is
$\omega$-almost surely independent of $\omega$ (see
e.g.~\cite{Kirsch:1989,MR94h:47068}); we
denote it by $\Sigma_{\lambda}$ .\\
Our assumptions on $q_{\lambda,\omega}$ imply that $\Sigma_{\lambda}$
is bounded below. Define $E_{\lambda}:=\inf\Sigma_{\lambda}$.\\
The goal of the present paper is to study the nature of the spectrum
of $H_{\lambda,\omega}$ near $E_{\lambda}$. A result typical of the
class of results that we will prove is
\begin{Th}
  \label{thr:2}
  Assume $p$ is not constant and that the random variables
  $(\omega_\gamma)_{\gamma\in\Z^d}$ are uniformly distributed in the
  unit ball in $\R^d$.\\
  Then, there exists $\varepsilon_0>0$ and $\lambda_0>0$ such that,
  for a generic single site potential $q$ such that
  $\|q\|_\infty\leq\varepsilon_0$, for $\lambda \in(0, \lambda_0]$,
  Anderson and strong dynamical localization near the bottom of the
  spectrum $E_\lambda$. Namely, there exist
  $E_{\lambda,1}>E_{\lambda}$ such that $H_{\lambda,\omega}$ has dense
  pure point spectrum on $[E_{\lambda}, E_{\lambda,1}]$ almost surely,
  and each eigenfunction associated to an energy in this interval
  decays exponentially as $|x| \rightarrow \infty$, and strong
  dynamical localization holds in the same region.
\end{Th}
\noindent For details on strong dynamical localization, we refer
to\cite{MR2002m:82035}.\\
When studying its spectral properties, an important feature of
$H_{\lambda,\omega}$ is that it depends non monotonically (see
e.g.~\cite{FN:08a}) on the random variables $(\omega_{\gamma})_{\gamma
  \in {\zed}^d}$, even if $q$ is assumed to be sign-definite. As each
of the random variables $(\omega_{\gamma})_{\gamma \in {\zed}^d}$ is
multidimensional, there cannot be a real monotonicity. Nevertheless,
we exhibit a set of assumptions on the single site potential $q$ and
on the random variables $(\omega_{\gamma})_{\gamma \in {\zed}^d}$ that
guarantee that, for sufficiently small disorder $\lambda$,
\begin{itemize}
\item there exists a neighborhood of $E_{\lambda}$ where
  $H_{\lambda,\omega}$ admits a Wegner estimate,
\item $H_{\lambda,\omega}$ exhibits a Lifshitz tails at $E_{\lambda}$.
\end{itemize}
It is well known that such results then entail Anderson and dynamical
localization near $E_\lambda$ (see e.g.~\cite{MR2002m:82035}). \\
Our assumptions are presumably not optimal; we show that they hold for
a small generic $q$. We need to assume some regularity for the
distribution of the random variables. As they are multi-dimensional,
absolute continuity with respect to the $d$-dimensional Lebesgue
measure is not necessary; actually, they can be concentrated on
subsets of dimension one (see section~\ref{sec:wegner-estimate-1}). As
for the support of the single site random variable, they can have a
wide variety of shapes but need to satisfy a type of strict convexity
condition at certain points; we refer to
section~\ref{sec:valid-assumpt-h.1} for more details.
\par Due to the non monotonicity of $H_{\lambda,\omega}$, few rigorous
results are known for the random displacement model in dimension
larger than 1. \\
For the one-dimensional displacement model, localization at all
energies was proven in~\cite{Buschmann:2001} and, with different
methods and, under more general assumptions, in~\cite{Damanik:2002}.
These proofs establish the Wegner estimate using two-parameter
spectral averaging and use lower bounds on the Lyapunov exponent to
replace the Lifshitz tails behavior. \\
For the multi-dimensional random displacement model, the only
available result on localization prior to the present paper
was~\cite{Klopp:1993} establishing the existence of a localized region
for the semi-classical operator $-h^2 \Delta+p +q_{\lambda,\omega}$
when $h$ is sufficiently small. The Wegner estimate was established
through a careful analysis of quantum tunneling. The Lifshitz tails
behavior was neither proved nor used in the energy region under
consideration, because of the semi-classical regime, the model is in a
large disorder regime.\\
It has been discovered recently that, for random displacement models,
Lifshitz tails need not hold (see~\cite{BaLoSto:08,FN:09}).\\
Related to the study of the occurrence of the Lifshitz tails, an
important point is the study of the infimum of the almost sure
spectrum and, in particular of the finite volume configurations of the
random parameter, if any, that give rise to the same ground state
energy. Such a study for non monotonous models has been undertaken
recently in~\cite{MR2430638,FN:08a}. In the present paper, we give an
analysis of those configuration in the small displacement case.
\section{The main results}
\label{sec:main-results}
For $n\geq0$, let $\Lambda_n=[-n-1/2,n+1/2]^d$. For
$(\omega_\gamma)_{\gamma\in{\zed}}$, define the differential
expression
\begin{equation}
  \label{eq:3}
  H_{\lambda,\omega,n}=-\Delta+ p+\sum_{\beta \in
    (2n+1){\zed}^d}\;\; \sum_{\gamma \in {\zed}^d/(2n+1){\zed}^d}
  q(x-\beta-\gamma-\lambda \omega_{\gamma}).
\end{equation}
Let $H^P_{\lambda,\omega,n}$ be restriction of $H_{\lambda,\omega,n}$
to the cube $\Lambda_n$ with periodic boun\-dary
conditions. $H_{\lambda,\omega,n}^P$ has only discrete spectrum and is
bounded from below.  For $E \in \er$, the {\it integrated density of
  states} is, as usual, defined by
\begin{equation*}
  N_{\lambda}(E)=\lim_{n\to+\infty}\frac1{(2n+1)^d}
  \#\{\text{eigenvalues of }H_{\lambda,\omega,n}^P\text{ in
  }(-\infty,E]\}.
\end{equation*}
We refer to~\cite{Kirsch:1989,MR94h:47068} for details on this
function and the proofs of various standard results.
\subsection{The assumptions}
\label{sec:assumptions}
We now state our assumptions on the random potential. Therefore, we
introduce the periodic operator obtained by shifting all the single
site potentials by exactly the same amount i.e. for $\zeta\in K$ (see
assumption (H.0.2)), let
\begin{equation}
  \label{eq:6}
  H_\zeta=H_{\lambda,\overline{\zeta}}=
  -\Delta+p+\sum_{\gamma\in\Z^d}q(x-\gamma-\lambda\zeta).
\end{equation}
Here and in the sequel, $\overline{\zeta}$ denotes the constant vector
with entries all equal to $\zeta$ i.e. $\overline{\zeta}=
(\zeta)_{\gamma\in\Z^d}$.\\
The spectrum of the $\Z^d$-periodic operator $H_\zeta$ is purely
absolutely continuous; it is a union of intervals (see
e.g.~\cite{MR58:12429c}). Let $E(\lambda,\zeta)$ be the infimum of
this spectrum. As $E(\lambda,\zeta)$ is the bottom of the spectrum of
the periodic operator $H_\zeta$, we know that it is a simple Floquet
eigenvalue associated to the Floquet quasi-momentum $\theta=0$ (see
section~\ref{sec:floquet-theory} for more details); hence, it is a
twice continuously differentiable function of $\zeta$.\\
We assume that
\begin{description}
\item[(H.1.1)] there exits $\lambda_0>0$ such that, for
  $\lambda\in(0,\lambda_0)$, there exists a unique point
  $\zeta(\lambda)\in K$ so that
  \begin{equation*}
    E(\lambda,\zeta(\lambda))=\min_{\zeta\in K}E(\lambda,\zeta);
  \end{equation*}
\item[(H.1.2)] there exists $\alpha_0>0$ such that, for
  $\lambda\in(0,\lambda_0)$ and $\zeta\in K$, one has
  \begin{equation}
    \label{eq:10}
    \nabla_\zeta E(\lambda,\zeta(\lambda))\cdot(\zeta-\zeta(\lambda))
    \geq\alpha_0\,\lambda\,|\zeta-\zeta(\lambda)|^2.
  \end{equation}
\end{description}
In section~\ref{sec:valid-assumpt-h.1}, we discuss concrete conditions
on $p$, $q$ and $K$ that ensure that assumption (H.1) is valid. We now
turn to our main results.
\subsection{The results}
\label{sec:results}
We start with a description of the realizations of the random
potential where the infimum of the almost sure spectrum is
attained. Then, we state our results on Lifshitz tails, a Wegner
estimate and the result on localization.
\subsubsection{The infimum of the almost sure spectrum}
\label{sec:infimum-spectrum}
Of course, as $\Sigma_\lambda$ is the almost sure spectrum, almost all
realizations have their infimum as the infimum of the spectrum. The
realizations we are interested in are those that attain this infimum
when restricted to a finite volume. In the present paper, we construct
these restrictions using periodic boundary conditions, actually
considering periodic realizations of the random
potential. In~\cite{MR2430638,BaLoSto:08,FN:08a,FN:09}, the
restrictions were performed using Neumann boundary
conditions.\\
We define periodic configurations of the random potential. Fix
$n\geq0$ and, for
$(\omega_\gamma)_{\gamma\in{\zed}^d/(2n+1){\zed}^d}$, consider the
differential operator $H_{\lambda,\omega,n}$ defined by~\eqref{eq:3}
with domain $\mathcal{H}^2(\R^d)$.  It is $(2n+1)\Z^d$-periodic; let
$E_0^n(\lambda\omega)$ be its ground state energy i.e. the infimum of
its spectrum. \\
One has
\begin{Th}
  \label{thr:1}
  Under assumptions (H.0) and (H.1), there exists $\lambda_0>0$ such
  that, for any $n\geq0$, for $\lambda\in(0,\lambda_0]$, on
  $K^{(2n+1)^d}$, the function $\omega\mapsto E_0^n(\lambda\omega)$
  reaches its infimum $E(\lambda,\zeta(\lambda))$ at a single point,
  the point
  $\omega=(\zeta(\lambda))_{\gamma\in{\zed}^d/(2n+1){\zed}^d}$.
\end{Th}
\noindent So, when it comes to finding the ``ground state'' of our
random system, for small $\lambda$, the Hamiltonian behaves as if it
were monotonous in the random variables $(\omega_\gamma)_\gamma$.\\
By the standard characterization of the almost sure spectrum in terms
of the spectra of the periodic approximations (see
e.g.~\cite{MR94h:47068}), for $\lambda$ sufficiently small, one has
that $E_\lambda=\inf\Sigma_\lambda=E(\lambda,\zeta(\lambda))$.
\subsubsection{The Lifshitz tails}
\label{sec:lifshitz-tails}
As a consequence of the determination of the minimum, we obtain
\begin{Th}
  \label{t2}
  Under assumptions (H.0) and (H.1), there exists $\lambda_0>0$ such
  that for all $\lambda\in(0,\lambda_0]$,
  \begin{equation*}
    \lim_{E \rightarrow E_{\lambda}} \frac{\log
      |\log(N_{\lambda}(E)-N_{\lambda}(E_{\lambda})|}
    {\log(E-E_{\lambda})}\leq-\frac{d}{2} 
  \end{equation*}
  Moreover, if the common distribution of the random variables
  $(\omega_\gamma)_\gamma$ is such that, for all $\lambda$,
  $\varepsilon$ and $\delta$ positive sufficiently small, one has
  \begin{equation*}
    \mathbb{P}(\{|\omega_0-\zeta(\lambda)|\leq\varepsilon\})
    \geq e^{-\varepsilon^{-\delta}},
  \end{equation*}
  then
  \begin{equation*}
    \lim_{E \rightarrow E_{\lambda}} \frac{\log
      |\log(N_{\lambda}(E)-N_{\lambda}(E_{\lambda})|}
    {\log(E-E_{\lambda})}=-\frac{d}{2} 
  \end{equation*}
\end{Th}
\noindent The Lifshitz tail behavior is well known for monotonous
alloy type models. It has also been discovered recently that, for
general displacement or non monotonous alloy type models, this
behavior need not hold (see~\cite{BaLoSto:08,FN:08a,FN:09}).
\subsubsection{The Wegner estimate}
\label{sec:wegner-estimate-1}
A Wegner estimate is an estimate on the probability that a restriction
of the random Hamiltonian to a cube admits an eigenvalue in a fixed
energy interval. Clearly, the estimate should grow with the size of
the cube and decrease with the length of the interval in which one
looks for eigenvalues.\\
The restrictions we choose are the periodic ones i.e those defined at
the beginning of section~\ref{sec:results}. We assume that
\begin{description}
\item[(H.2)] There exists $C>0$ such that, for $\lambda$ sufficiently
  small, one has $E_\lambda\leq E_0-\lambda/C$.
\end{description}
Clearly, Theorem~\ref{thr:1} shows that this assumption is a
consequence of assumptions (H.0) and (H.1).\\
For the alloy type models, it is well known that a Wegner estimate
will hold only under a regularity assumption. We now turn to the
corresponding assumption for our displacement model. We keep the
notations of section~\ref{sec:infimum-spectrum}. Consider the polar
decomposition of the random variable $\omega_0$, say
$\omega_0=r(\omega_0)\sigma(\omega_0)$. For $\sigma\in{\mathbb
  S}^{d-1}$, define $r_\sigma(\omega_0)$, the random variable
$r(\omega_0)$ conditioned on $\sigma(\omega_0)=\sigma$.\\
We assume that
\begin{description}
\item[(H.3)] for almost all $\sigma\in{\mathbb S}^{d-1}$, the
  distribution of $r_\sigma(\omega_0)$ admits a density with respect
  to the Lebesgue measure, say, $h_\sigma$ that itself is absolutely
  continuous with respect to the Lebesgue measure; moreover, one has
  \begin{equation}
    \label{eq:23}
    \text{ess-sup}_{\sigma\in{\mathbb S}^{d-1}}\|h_\sigma'\|_\infty<+\infty.
  \end{equation}
\end{description}
\begin{Rem}
  \label{rem:1}
  Assumption (H.3) will hold for example if
  \begin{itemize}
  \item the random variable admit a density that is continuously
    differentiable on its support;
  \item the random variable is supported on a submanifold of dimension
    $1\leq d'\leq d$, and on this submanifold, it admits continuously
    differentiable density.
  \end{itemize}
\end{Rem}
We prove
\begin{Th}\label{t3}
  Under assumptions (H.0), (H.2) and (H.3), for any $\nu\in(0,1)$,
  there exists $\lambda_0>0$ such that, for $\lambda\in(0,\lambda_0]$,
  there exists $C_\lambda>0$ such that, for all $E\in[E_{\lambda},
  E_{\lambda}+\lambda/C]$ and $\varepsilon >0$ such that
  \begin{equation}
    \label{eq:18}
    \mathbb{P}(\text{dist}(\sigma(H_{\lambda,\omega,n}^P),E)\leq\varepsilon)
    \leq C_\lambda\varepsilon^{\nu}n^d .
  \end{equation}
\end{Th}
\noindent The result is essentially a quite simple consequence of
Theorem 6.1 of~\cite{MR1934351}; the modifications are indicated in
section~\ref{sec:wegner-estimate}.\\
In the case of monotonous random operators, under our smoothness
assumptions for the distribution of the random variables, the
estimate~(\ref{eq:18}) can be improved in the sense that the power
$\nu$ can be taken equal to 1 (see~\cite{MR2362242}). It seems
reasonable to think that the same holds true for most non monotonous
models; to our knowledge, no proof of this fact exists.\\
A Wegner estimate of the type~(\ref{eq:18}) implies a minimal
regularity for $N_\lambda$, the integrated density of states of
$H_{\lambda,\omega}$ in the low energy region. Indeed, one proves
\begin{Cor}
  \label{cor:1}
  Under the assumptions of Theorem~\ref{t3}, for any $\nu\in(0,1)$,
  the integrated density of states $N_\lambda$ is $\nu$-H{\"o}lder
  continuous is the region $[E_{\lambda},E_{\lambda}+\lambda/C]$
  defined in Theorem~\ref{t3}.
\end{Cor}
\subsubsection{Localization}
\label{sec:localization}
Once Theorems \ref{t2} and \ref{t3} are proved, localization follows
by the now standard multiscale argument (see
e.g.~\cite{MR2002m:82035})
\begin{Th}
  \label{t1}
  Under assumptions (H.0), (H.1) and (H.3), there exists $\lambda_0>0$
  such that, for $\lambda \in(0, \lambda_0]$, Anderson and strong
  dynamical localization near the bottom of the spectrum. Namely,
  there exist $E_{\lambda,1}>E_{\lambda}$ such that
  $H_{\lambda,\omega}$ has dense pure point spectrum on $[E_{\lambda},
  E_{\lambda,1}]$ almost surely, and each eigenfunction associated to
  an energy in this interval decays exponentially as $|x| \rightarrow
  \infty$, and strong dynamical localization holds in the same region.
\end{Th}
\noindent We omit the details of the proofs of this result. We only
note that the Combes-Thomas estimate and the decomposition of
resolvents in the multiscale argument work for the random displacement
model in the same way as for alloy type models.
\subsection{The validity of assumption (H.1)}
\label{sec:valid-assumpt-h.1}
Let us now describe some concrete conditions on $q$ and $K$ that
ensure that assumption (H.1) does hold. Let $H_0=H_{\lambda,0}$ be
defined by~\eqref{eq:6} for $\zeta=0$. The spectrum of this operator
is purely absolutely continuous; it is a union of intervals (see
e.g.~\cite{MR58:12429c}). Let $E_0$ be the infimum of this spectrum
and $\varphi_0$ be the solution to the following spectral problem
\begin{equation}
  \label{eq:12}
  \begin{cases}
    H_0\varphi_0=E_0\varphi_0,\\\forall\gamma\in\Z^d,\
    \varphi_0(x+\gamma)=\varphi_0(x).
  \end{cases}
\end{equation}
This solution is unique up to a constant; it can be chosen positive
and normalized (see~\cite{MR89b:35127,MR80m:81085}). We will then call
it the ground state for $H_0$.\\
Recall that $K$ is the essential support of the random
variables $(\omega_\gamma)_\gamma$; thus $K\subset\R^d$. \\
We prove
\begin{Pro}
  \label{pro:1}
  Assume that $K$ is
  \begin{itemize}
  \item either a convex set with $C^2$-boundary such that all its
    principal curvatures are positive at all points,
  \item or the boundary of such a convex set,
  \end{itemize}
  and that
  \begin{equation}
    \label{eq:11}
    v(q):=-\int_{\R^d}\nabla q(x)|\varphi_0(x)|^2dx\not=0,
  \end{equation}
  Then, assumption (H.1) holds.
\end{Pro}
\noindent For a fixed periodic potential $p$ that is not constant, by
perturbation theory, it is not difficult to see that
condition~\eqref{eq:11} is satisfied for a generic small $q$. Indeed,
if $\psi_0$ is the ground state for $-\Delta+p$ (in the sense defined
above), as $\psi_0$ is positive, its modulus is constant if and only
if it is constant. In which case, the eigenvalue
equation~\eqref{eq:12} tells us that $p$ is constant, identically
equal to $E_0$. So we may assume that $\psi_0$ is not constant, one
can then find $q$ smooth and compactly supported such
that~\eqref{eq:11} holds. Indeed, by integration by parts,
\begin{equation*}
  w(q):=\int_{\R^d}\partial_i
  q(x)\psi^2_0(x)dx=2\int_{\R^d}q(x)\psi_0(x)\partial_i
  \psi_0(x)dx 
\end{equation*}
which vanishes for all smooth compactly supported functions if and
only if $\partial_i\psi_0$ vanishes identically. Hence, $w(q)$
vanishes for all $q$ small, smooth and compactly supported if and only
if $\psi_0$ is a constant (as $q\mapsto w(q)$ is linear).\\
As $\varphi_0$ is the ground state for the operator
$-\Delta+p+\sum_{\gamma}q(\cdot-\gamma)$ and this ground state is a
real analytic function of the potential $q$, the difference
$\psi_0-\varphi_0$ is small for $q$ small. So, if we pick $q_0$ such
that $w(q_0)\not=0$, for $\varepsilon$ small and $q=\varepsilon q_0$,
we know that $v(q)$ does not vanish i.e.~\eqref{eq:11} is satisfied.
\par By Proposition~\ref{pro:1} and Remark~\ref{rem:1}, it is clear
now that Theorem~\ref{thr:2} is a consequence of Theorem~\ref{t1}.
\par Let us now give another assumption on $K$ under which (H.1)
holds. We prove
\begin{Pro}
  \label{pro:2}
  Assume that~\eqref{eq:11} is satisfied and that the set $K$
  satisfies that, there exists $\varepsilon>0$ and $\zeta_0\in K$,
  such that, for all $\zeta\in K$ and $|v-v(q)|<\varepsilon$, one has
  \begin{equation*}
    v\cdot(\zeta-\zeta_0)\geq0.
  \end{equation*}
  Then, assumption (H.1) holds. Moreover, for $\lambda$ small, the
  minimum $\zeta(\lambda)$ satisfies $\zeta(\lambda)=\zeta_0$.
\end{Pro}
\noindent Before we proceed to the proofs of Propositions~\ref{pro:1}
and~\ref{pro:2}, let us compare our setting to the one studied
in~\cite{MR2430638,BaLoSto:08,FN:09}. In those studies,
assumption~\eqref{eq:11} but also assumption~(H.1) are not
fulfilled. Indeed, there, $p$ and $q$ are assumed to be reflection
symmetric with respect to the coordinate planes i.e.  for any
$\sigma=(\sigma_1,\dots,\sigma_d)\in\{0,1\}^d$ and any
$x=(x_1,\dots,x_d)\in\R^d$,
\begin{equation*}
  q(x_1,\dots,x_d)=q((-1)^{\sigma_1}x_1,\dots,(-1)^{\sigma_d}x_d).
\end{equation*}
Hence, the potential $\D p(\cdot)+\sum_\gamma q(\cdot-\gamma)$ and the
ground state $\varphi_0$ satisfy the same reflection symmetry. This
implies that
\begin{equation*}
  \int_{\R^d}\nabla q(x)|\varphi_0(x)|^2dx=
  -\int_{\R^d}\nabla q(x)|\varphi_0(x)|^2dx=0.
\end{equation*}
The fact that, in the setting of~\cite{MR2430638,BaLoSto:08,FN:09},
assumption (H.1.1) is not satisfied is seen directly from those papers
as the ground state of the periodic operator $H_\zeta$ reaches its
minimum at $2^d$ values as soon as $K$ is reflection symmetric.
\subsection{The proofs of Propositions~\ref{pro:1} and~\ref{pro:2}}
\label{sec:proof-prop-refpr}
Consider the mapping $\zeta\mapsto F(\lambda,\zeta)=\lambda^{-1}
E(\lambda,\zeta)$ on some large ball $B$ containing $K$. As
$E(\lambda,\zeta)$ is a simple Floquet eigenvalue associated to the
normalized Floquet eigenvector $\varphi_0(\lambda,\zeta,0)$ (see
section~\ref{sec:floquet-theory}), we can compute the gradient of $F$
in the $\zeta$-variable using the Feynman-Hellmann formula to obtain
\begin{equation*}
  \nabla_\zeta F(\lambda,\zeta)=-\int_{\R^d}\nabla q(x-\lambda\zeta)
  |\varphi_0(\lambda,\zeta,0;x)|^2dx.
\end{equation*}
Hence,
\begin{equation}
  \label{eq:14}
  \sup_{\zeta\in B}|\nabla_\zeta F(\lambda,\zeta)-v(q)|\vers_{\lambda\to0}0
\end{equation}
\begin{proof}[Proof of Proposition~\ref{pro:1}]
  Assume first that $K$ is a convex set satisfying the assumptions of
  Proposition~\ref{pro:1}. Using the rectification theorem (see
  e.g.~\cite{MR2242407}), assumption~\eqref{eq:11} and
  equation~\eqref{eq:14} guarantee that, for $\lambda$ small, one can
  find a $C^2$-diffeomorphism, say $\Psi_\lambda$, from $B$ to
  $\Psi_\lambda(B)$ such that $|\Psi_\lambda-\text{Id}|_{C^2}\to0$
  when $\lambda\to 0$ and
  \begin{equation*}
    \nabla_\zeta(F(\lambda,\Psi_\lambda(\zeta)))=v(q).
  \end{equation*}
  Now, assume that $K$ is a convex set with a $C^2$-boundary having
  all its principal curvatures positive at all points. Then, for
  $\lambda$ small, the set $K_\lambda=\Psi^{-1}_\lambda(K)$ also is
  convex with a $C^2$-boundary having all its principal curvatures
  positive at all points; moreover, the curvatures are bounded away
  from $0$ independently of $\lambda$ for $\lambda$ small.\\
  On the convex set $K_\lambda$, the affine function
  $G(\zeta):=F(\lambda,\Psi_\lambda(\zeta))=v(q)\cdot\zeta+C_\lambda$
  reaches its infimum at a single point, say
  $\tilde\zeta(\lambda)=\Psi^{-1}_\lambda(\zeta(\lambda))$,
  $\zeta(\lambda)\in\partial K$.\\
  Hence, we have that, for $\zeta\in K\setminus\{\zeta(\lambda)\}$,
  $F(\lambda,\zeta)>F(\lambda,\zeta(\lambda))$. The convexity of $K$
  ensures that, for $\zeta\in K\setminus\{\zeta(\lambda)\}$, one has
  \begin{equation}
    \label{eq:8}
    \nabla_\zeta F(\lambda,\zeta(\lambda))\cdot(\zeta-\zeta(\lambda))\geq0.
  \end{equation}
  Indeed, as $K$ is convex, for $\nu\in(0,1)$ and $\zeta\in
  K\setminus\{\zeta(\lambda)\}$, one has
  $\zeta_\nu=\nu\zeta+(1-\nu)\zeta(\lambda)\in
  K\setminus\{\zeta(\lambda)\}$; thus
  $F(\lambda,\zeta_\nu)>F(\lambda,\zeta(\lambda))$. Taking the right
  hand side derivative of $\nu\mapsto F(\lambda,\zeta_\nu)$ at $\nu=0$
  yields~\eqref{eq:8}.\\
  The strict convexity of $K$, guaranteed by the positivity of the
  principal curvatures of $\partial K$, ensures that, for $\zeta\in
  K\setminus\{\zeta(\lambda)\}$, one has
  \begin{equation}
    \label{eq:13}
    \nabla_\zeta F(\lambda,\zeta(\lambda))\cdot(\zeta-\zeta(\lambda))>0.
  \end{equation}
  Indeed, assume that for some $\zeta_0\in
  K\setminus\{\zeta(\lambda)\}$, \eqref{eq:13} is not satisfied i.e
  $\nabla_\zeta F(\lambda,\zeta(\lambda))\cdot(\zeta_0-\zeta(\lambda))
  =0$. As $K$ is strictly convex, $K$ contains a cone of the form
  $\{\zeta(\lambda)+ r(\zeta_0-\zeta(\lambda))+ rw;\ \|w\|\leq1,\
  r\in[0,r_0]\}$ for some small $r_0>0$. Picking $w$ such that
  $\nabla_\zeta F(\lambda,\zeta(\lambda)) \cdot w<0$, one constructs
  $\zeta'\in K$ such that $\nabla_\zeta F(\lambda,\zeta(\lambda))
  \cdot(\zeta'-\zeta(\lambda))<0$ which contradicts~\eqref{eq:8}.\\
  To show~(\ref{eq:10}), it suffices to show that, for $\zeta\in K$,
  \begin{equation}
    \label{eq:4}
    \nabla_\zeta F(\lambda,\zeta(\lambda))\cdot(\zeta-\zeta(\lambda))
    \geq\frac1{C_0}|\zeta-\zeta(\lambda)|^2.
  \end{equation}
  Let $H_\lambda$ be the hyperplane orthogonal to $\nabla_\zeta
  F(\lambda,\zeta(\lambda))$ at $\zeta(\lambda)$. It intersects $K$ at
  $\zeta(\lambda)$ and $K$ is contained in one of the half-spaces
  defined by this hyperplane. Thus, the hyperplane is tangent to $K$
  at $\zeta(\lambda)$ (see e.g.~\cite{MR2311920}). Hence, there exists
  $\alpha_0>0$ such that, for $\zeta\in K$, one has
  \begin{equation}
    \label{eq:15}
    \nabla_\zeta F(\lambda,\zeta(\lambda))\cdot(\zeta-\zeta(\lambda))
    \geq\alpha_0\, d(\zeta,H_\lambda)^2
  \end{equation}
  where $d(\zeta,H_\lambda)$ denotes the distance from $\zeta$ to
  $H_\lambda$. The constant $\alpha_0$ can be chosen independent of
  $\lambda$ for $\lambda$ small as the principal curvatures of
  $\partial K$ are uniformly positive. Now, if $u=\Vert\nabla_\zeta
  F(\lambda,\zeta(\lambda))\Vert^{-1}\nabla_\zeta F(\lambda,\zeta
  (\lambda))$, for $\zeta\in K$ as $K$ is compact, one has
  \begin{equation}
    \label{eq:5}
    \begin{split}
      \nabla_\zeta
      F(\lambda,\zeta(\lambda))\cdot(\zeta-\zeta(\lambda))&=
      \Vert\nabla_\zeta F(\lambda,\zeta(\lambda))\Vert\,
      [u\cdot(\zeta-\zeta(\lambda))]\\&\geq
      \alpha_0\,[u\cdot(\zeta-\zeta(\lambda))]^2.
    \end{split}
  \end{equation}
  As $|\zeta-\zeta(\lambda)|^2=d(\zeta,H_\lambda)^2+
  [u\cdot(\zeta-\zeta(\lambda))]^2$, the lower bounds~(\ref{eq:15})
  and~\eqref{eq:5} imply~(\ref{eq:4}).\\
  To deal with the case when $K$ is the boundary of a convex set, we
  only need to do the analysis done above for the convex hull of $K$
  and notice that the minimum is attained on $K$ the boundary of this
  convex hull.\\
  This completes the proof of Proposition~\ref{pro:1}.
\end{proof}
\begin{proof}[Proof of Proposition~\ref{pro:2}]
  By assumption, for $\zeta\in K$ and $|v-v(q)|<\varepsilon$, one has
  $v\cdot(\zeta-\zeta_0)\geq0$. Hence, as $K$ is compact, there exists
  $c>0$ such that, for all $\zeta\in K$ and $|v-v(q)|<\varepsilon/2$,
  one has
  \begin{equation}
    \label{eq:16}
    v\cdot(\zeta-\zeta_0)\geq c|\zeta-\zeta_0|.
  \end{equation}
  Let $B$ be a closed ball centered in $\zeta_0$ such that $K\subset
  B$. By~(\ref{eq:14}), (\ref{eq:16})~implies that, for $\lambda$
  sufficiently small, for all $\tilde\zeta\in B$ and $\zeta\in K$, one
  has
  \begin{equation*}
    \nabla_\zeta F(\lambda,\tilde\zeta)\cdot(\zeta-\zeta_0)\geq
    c|\zeta-\zeta_0|.
  \end{equation*}
  Hence,
  \begin{equation*}
    \begin{split}
    F(\lambda,\zeta)-F(\lambda,\zeta_0)&=\int_0^1\nabla_\zeta
    F(\lambda,\zeta_0+t(\zeta-\zeta_0))
    \cdot(\zeta-\zeta_0)dt\\&\geq c|\zeta-\zeta_0|.      
    \end{split}
  \end{equation*}
  So $\zeta_0$ is the unique minimum of $\zeta\mapsto
  F(\lambda,\zeta)$ in $K$ i.e. for $\lambda$ sufficiently small,
  $\zeta(\lambda)=\zeta_0$. Using again the boundedness of $K$, we get
  the estimate~(\ref{eq:10}) of assumption (H.1). This completes the
  proof of Proposition~\ref{pro:2}.
\end{proof}
\section{The reduction to a discrete model}
\label{sec:caract-ground-state}
In this section, we prove the results announced in
section~\ref{sec:infimum-spectrum}. Therefore, we will use the Floquet
decomposition for periodic operators to reduce our operator to some
discrete model in the way it was done
in~\cite{Klopp:1999,Ghribi:2006}.
\subsection{Floquet theory}
\label{sec:floquet-theory}
Pick $\zeta \in K$ and let $H_\zeta$ be the $\Z^d$-periodic operator
defined by~\eqref{eq:6}. For $\theta\in\T^*:={\er}^d/(2\pi\Z^d)$ and
$u \in {\mathcal S}({\er}^d)$, the Schwartz space of rapidly decaying
functions, following~\cite{MR58:12429c}, we define
\begin{equation*}
  (Uu)(\theta,x)=\sum_{\gamma \in {\zed}^d} e^{i \gamma \cdot \theta} u(x-\gamma)
\end{equation*}
which can be extended as a unitary isometry from $L^2({\er}^d)$ to
${\mathcal H}:=L^2(K_0\times\T^*)$ where $K_0=(-1/2,1/2]^d$ is the
fundamental cell of ${\zed}^d$. The inverse of $U$ is given by
\begin{equation*}
  \mbox{for} \; v \in {\mathcal H}, \;
  (U^*v)(x)=\frac{1}{\mbox{Vol}(\T^*)} \int_{\T^*}
  v(\theta,x) d\theta.
\end{equation*}
As $H_{\lambda,\overline{\zeta}}$ is ${\zed}^d$-periodic,
$H_{\lambda,\overline{\zeta}}$ admits the Floquet decomposition
\begin{equation*}
  UH_{\lambda,\overline{\zeta}}U^*=\int^{\oplus}_{\T^*}
  H_{\lambda,\overline{\zeta}}(\theta) d\theta  
\end{equation*}
where $H_{\lambda,\overline{\zeta}}(\theta)$ is the differential
operator $H_{\lambda,\overline{\zeta}}$ acting on ${\mathcal
  H}_{\theta}$ with domain ${\mathcal H}_{\theta}^2$ where
\begin{itemize}
\item for $v\in\R^d$, $\tau_v:\ L^2(\R^d)\to L^2(\R^d)$ denotes the
  ``translation by $v$'' operator i.e for $\varphi\in L^2(\R^d)$ and
  $x\in\R^d$, $(\tau_v\varphi)(x)=\varphi(x-v)$;
\item ${\mathcal D}'_\theta$ is the space $\theta$-quasi-periodic
  distribution in $\R^d$ i.e the space of distributions $u\in{\mathcal
    D}'(\R^d)$ such that, for any $\gamma\in\Z^d$, we have
  $\D\tau_\gamma u=e^{-i\gamma\cdot \theta}u$. Here $\theta\in\T^*$;
\item $\mathcal{H}^k_{\text{loc}}(\R^d)$ is the space of distributions that
  locally belong to $\mathcal{H}^k(\R^d)$ and we define ${\mathcal
    H}^k_\theta=\mathcal{H}^k_{\text{loc}}(\R^d)\cap{\mathcal D}'_\theta$;
\item for $k=0$, we define ${\mathcal H}_\theta={\mathcal H}^0_\theta$
  and identify it with $L^2(K_0)$; equipped with the $L^2$-norm over
  $K_0$, it is a Hilbert space; the scalar product will be denoted by
  $\langle\cdot ,\cdot \rangle_\theta$.
\end{itemize}
We know that $H_{\lambda,\overline{\zeta}}(\theta)$ is self-adjoint
and has a compact resolvent; hence its spectrum is discrete. Its
eigenvalues repeated according to multiplicity, called Floquet
eigenvalues of $H_{\lambda,\overline{\zeta}}$, are denoted by
\begin{equation*}
  E_0(\lambda,\zeta,\theta) \leq
  E_1(\lambda,\zeta,\theta) \leq \cdots\leq
  E_n(\lambda,\zeta,\theta)\to+\infty.
\end{equation*}
The functions $((\lambda,\zeta,\theta)\mapsto
E_{n}(\lambda,\zeta,\theta))_{n \in \en}$ are Lipschitz-continuous in
the variable $\theta$; they are even analytic in
$(\lambda,\zeta,\theta)$ when they are simple eigenvalues.\\
Define $\varphi_n(\lambda,\zeta,\theta)$ to be a normalized
eigenvector associated to the eigenvalue $E_n(\lambda,\zeta,\theta)$.
The family $(\varphi_n(\lambda,\zeta,\theta))_{n\geq0}$ is chosen so
as to be a Hilbert basis of $\mathcal{H}_\theta$. If
$E_n(\lambda_0,\zeta_0,\theta_0)$ is a simple eigenvalue, the function
$(\lambda,\zeta,\theta)\mapsto \varphi_n (\lambda,\zeta,\theta)$ is
analytic near
$(\lambda_0,\zeta_0,\theta_0)$.\\
It is well known (see e.g.~\cite{MR89b:35127}) that, for given
$\lambda$ and $\zeta$, the eigenvalue $E_0(\lambda,\zeta,\theta)$
reaches its minimum at $\theta=0$, and that it is simple for $\theta$
small.
\subsection{The reduction procedure}
\label{sec:reduction-procedure}
Recall that the $(\varphi_n(\lambda,\zeta,\theta))_{n \geq 0}$ are the
Floquet eigenvectors of $H_{\lambda,\overline{\zeta}}$. Let
$\Pi_{\lambda,\zeta,0}(\theta)$ and $\Pi_{\lambda,\zeta,+}(\theta)$
respectively denote the orthogonal projections in ${\mathcal
  H}_{\theta}$ on the vector spaces respectively spanned by
$\varphi_0(\lambda,\zeta,\theta)$ and $(\varphi_n
(\lambda,\zeta,\theta))_{n\geq1}$. Obviously, these projectors are
mutually orthogonal and their sum is the identity for any $\theta \in
\T^*$.\\
Define $\Pi_{\lambda,\zeta,\alpha}=U^*\Pi_{\lambda,\zeta,\alpha}
(\theta)U$ where $\alpha \in \{0, + \}$.  $\Pi_{\lambda,\zeta,\alpha}$
is an orthogonal projector on $L^2({\er}^d)$ and, for $\gamma \in
{\zed}^d$, we have $\tau_{\gamma}^*\Pi_{\lambda,\zeta,\alpha}
\tau_{\gamma}=\Pi_{\lambda,\zeta,\alpha}$. It is clear that that
$\Pi_{\lambda,\zeta,0}+\Pi_{\lambda,\zeta,+}=Id_{L^2({\er}^d)}$ and
$\Pi_{\lambda,\zeta,0}$ and $\Pi_{\lambda,\zeta,+}$ are mutually
orthogonal. For $\alpha \in \{0, + \}$, we set ${\mathcal
  E}_{\lambda,\zeta,\alpha}=\Pi_{\lambda,\zeta,\alpha}(L^2({\er}^d))$. These
spaces are invariant under translations by vectors in $\Z^d$ and
${\mathcal E}_{\lambda,\zeta,0}$ is of finite energy
(see~\cite{Klopp:1999}).\\
For $u \in L^2(\T^*)$, we define
\begin{equation*}
  P_{\lambda,\zeta}(u)=U^*(u(\theta) \varphi_0(\lambda,\zeta,\theta)).
\end{equation*}
The mapping $P_{\lambda,\zeta} :L^2(\T^*) \rightarrow {\mathcal
  E}_{\lambda,\zeta,0}$ defines a unitary equivalence
(see~\cite{Klopp:1999}); its inverse is given by
\begin{equation*}
  P_{\lambda,\zeta}^*(v)=\langle (Uv)(\theta),
  \varphi_0(\lambda,\zeta,\theta)\rangle,\quad v \in
  {\mathcal E}_{\lambda,0}.
\end{equation*}
One checks that
$P_{\lambda,\zeta}P_{\lambda,\zeta}^*=\Pi_{\lambda,\zeta,0}$ and
$P_{\lambda,\zeta}^*P_{\lambda,\zeta}=Id_{L^2(\T^*)}$.\\
The main result of this section is
\begin{Th}
  \label{thr:3}
  Under assumptions (H.0) and (H.1), there exists $C_0>0$ such that,
  for any $\alpha>0$, there exists $\lambda_0>0$ such that, for
  $\lambda\in(0,\lambda_0)$, for any $\zeta\in K$ and any
  $\omega=(\omega_\gamma)_{\gamma\in\Z^d}\in K^{\Z^d}$, one has
  \begin{equation}
    \label{eq:7}
    \begin{split}
      \frac1{C_0} &\left(P_{\lambda,\zeta} h^-_{\lambda,\omega,\zeta}
        P^*_{\lambda,\zeta}+ \Pi_{\lambda,\zeta,+} \right)
      \\&\hskip1cm\leq
      H_{\lambda,\omega}-E(\lambda,\zeta)\\&\hskip2cm\leq C_0
      \left(P_{\lambda,\zeta} h^+_{\lambda,\omega,\zeta}
        P^*_{\lambda,\zeta}+
        \tilde{H}_{\lambda,\overline{\zeta},+}\right)
    \end{split}
  \end{equation}
  where
  \begin{itemize}
  \item $\tilde{H}_{\lambda,\overline{\zeta},+}=
    (H_{\lambda,\overline{\zeta}}-E(\lambda,\zeta))\Pi_{\lambda,\zeta,+}$.
  \item $h^\pm_{\lambda,\omega,\zeta}$ is the random operator acting
    on $L^2(\T^*)$ defined by
    \begin{equation*}
      \begin{split}
        h^+_{\lambda,\omega,\zeta}&=C_0\,\varpi(\cdot)+
        \lambda\sum_{\gamma\in\Z^d} \left[v(\lambda,\zeta)
          \cdot(\omega_\gamma-\zeta)+
          C_0\,\alpha\,\|\omega_{\gamma}-\zeta\|^2\right] \Pi_\gamma\\
        h^-_{\lambda,\omega,\zeta}&=\frac1{C_0}\,\varpi(\cdot)+
        \lambda\sum_{\gamma\in\Z^d} \left[v(\lambda,\zeta)
          \cdot(\omega_\gamma-\zeta)-
          C_0\,\alpha\,\|\omega_{\gamma}-\zeta\|^2\right] \Pi_\gamma
      \end{split}
    \end{equation*}
  \item $\varpi$ is the multiplication operator by the function
    \begin{equation}
      \label{eq:27}
      \varpi(\theta)=\sum_{j=1}^d (1-\cos(\theta_j)),
    \end{equation}
  \item $\Pi_\gamma$ is the orthogonal projector on
    $e^{i\gamma\theta}$,
  \item the vector $v(\lambda,\zeta)$ is given by
    \begin{equation}
      \label{eq:9}
      v(\lambda,\zeta)=-\int_{{\er}^d} \nabla q(x-\lambda\zeta)
      |\varphi_0(\lambda,\zeta,0;x)|^2 dx
      =\frac1{\lambda}\nabla_\zeta E(\lambda,\zeta).
    \end{equation}
  \end{itemize}
\end{Th}
\noindent The proof of Theorem~\ref{thr:3} is the content of
section~\ref{sec:proof-theor-refthr:3}. We now use this result to
derive Theorem~\ref{thr:1} and~\ref{t2}.
\subsection{The characterization of the infimum of the almost sure
  spectrum}
\label{sec:char-infim-almost}
We now prove Theorem~\ref{thr:1}. Using Theorem~\ref{thr:3} for
$\zeta=\zeta(\lambda)$, we see that, for $\lambda$ sufficiently small,
for $\omega\in K^{(2n+1)^d}$, one has
\begin{equation}
  \label{eq:20}
  H_{\lambda,\omega,n}-E_\lambda
  \geq\frac1{C_0}\left(P_{\lambda,\zeta(\lambda)}
    h^-_{\lambda,\omega,\zeta(\lambda),\alpha,n}
    P^*_{\lambda,\zeta(\lambda)}+\Pi_{\lambda,\zeta(\lambda),+} \right).
\end{equation}
Using~\eqref{eq:10} and~(\ref{eq:9}), taking
$C_0\alpha\leq\alpha_0/2$, we get
\begin{equation*}
  C_0\,h^-_{\lambda,\omega,\zeta(\lambda),n} \geq \varpi(\cdot)+
  \frac{C_0 \lambda\,\alpha_0}2 \sum_{\beta \in
    (2n+1){\zed}^d}\;\; \sum_{\gamma \in {\zed}^d/(2n+1){\zed}^d}
  \|\omega_\gamma-\zeta(\lambda)\|^2\Pi_{\gamma+\beta}.
\end{equation*} 
As the spectrum of $h^-_{\lambda,\omega,\zeta(\lambda),n}$ is non
negative, the operator in the left hand side of~(\ref{eq:20}) is
clearly non negative; recall that $P_{\lambda,\zeta}
P_{\lambda,\zeta}^*+ \Pi_{\lambda,\zeta,+}=Id_{L^2}$,
$\Pi_{\lambda,\zeta,+}$ is an orthogonal projector and
$P_{\lambda,\zeta}^*$ is a partial unitary equivalence. \\
To prove Theorem~\ref{thr:1}, we will show that, if
$\omega\not=(\zeta(\lambda))_{\gamma\in{\zed}^d/(2n+1){\zed}^d}$,
then, there exists $c(\omega)>0$ such that
$h^-_{\lambda,\omega\zeta(\lambda),n}\geq c(\omega)$. Therefore,
recall that $h^-_{\lambda,\omega,\zeta(\lambda),n}$ is a periodic
operator so we can do its Floquet decomposition in the same way as in
section~\ref{sec:floquet-theory}. In the present case, as we deal with
a discrete model, the fiber operators will be finite dimensional
matrices (see e.g.~\cite{MR2003k:82050}); they can also be represented
as the operator $h^-_{\lambda,\omega,\zeta(\lambda),n}$ acting on the
finite dimensional space of linear combinations of the Dirac masses
$(\delta_{2\pi k/(2n+1)+\theta})_{k\in\Z^d/(2n+1)\Z^d}$; the
Floquet parameter $\theta$ belongs to $(2n+1)^{-1}\T^*$.\\
As $\varpi\geq0$ and $h^-_{\lambda,\omega,\zeta(\lambda),n}-\varpi
\geq0$, $0$ is in the spectrum of
$h^-_{\lambda,\omega,\zeta(\lambda),n}$ if and only if there exists
$\theta\in(2n+1)^{-1}\T^*$ and $v$, a linear combination of the Dirac
masses $(\delta_{2\pi k/(2n+1)+\theta})_{k\in\Z^d/(2n+1)\Z^d}$ (seen
as distributions on $\T^*$) such that $\varpi\cdot v=0$ and
$h^-_{\lambda,\omega,\zeta(\lambda),n}v=0$. Now, $\varpi\cdot v=0$ implies
that $\theta=0$ and $v=c\delta_0$. Hence,
$h^-_{\lambda,\omega,\zeta(\lambda),n}v=0$ implies that
\begin{equation*}
  \sum_{\gamma \in {\zed}^d/(2n+1){\zed}^d}
  \|\omega_\gamma-\zeta(\lambda)\|^2=0
\end{equation*}
i.e. $\omega=(\zeta(\lambda))_{\gamma\in{\zed}^d/(2n+1){\zed}^d}$.\\
So we see that the function $\omega\mapsto E_0^n(\lambda\omega)$
reaches its infimum only at the point
$\omega=(\zeta(\lambda))_{\gamma\in{\zed}^d/(2n+1){\zed}^d}$. This
completes the proof of Theorem~\ref{thr:1}. \qed
\subsection{The Lifshitz tails}
\label{sec:lifshitz-tail}
We now prove Theorem~\ref{t2}. Therefore, we again use the
reduction given by Theorem~\ref{thr:3}.\\
Fix $\zeta=\zeta(\lambda)$. First, the operators
$h^\pm_{\lambda,\omega,\zeta(\lambda)}$ are both standard discrete
Anderson models and, as such, admit each integrated density of states
that we denote by $N_r^\pm$. As we have seen in the previous section,
their spectra are contained in $\R^+$. \\
The inequality~(\ref{eq:7}) implies that, for $\lambda$ sufficiently
small and $E\in[0,1/C^2_0]$ where $C_0$ is the constant given in
Theorem~\ref{thr:3}, one has
\begin{equation*}
  N_r^+(E/C_0)\leq N_\lambda(E_\lambda+E)\leq N_r^-(C_0\,E).
\end{equation*}
Now, Theorem~\ref{t2} immediately follows from the existence of
Lifshitz tail for the Anderson models
$h^\pm_{\lambda,\omega,\zeta(\lambda)}$ (see e.g.~\cite{MR94h:47068})
which, in turn, follows from the facts that, for $\lambda$
sufficiently small, under our assumptions, if
$C_0\alpha\leq\alpha_0/2$, by~(\ref{eq:10}), the random variables
\begin{equation*}
  \omega^\pm_\gamma=[v(\lambda,\zeta(\lambda))
  \cdot(\omega_\gamma-\zeta(\lambda))]\pm C_0\alpha
  \|\omega_\gamma-\zeta(\lambda)\|^2 
\end{equation*}
are i.i.d, non negative, non trivial and $0$ belongs to their
support (see e.g.~\cite{MR94h:47068,Stollmann:2001}).\\
Now if the common distribution of the random variables
$(\omega_\gamma)_\gamma$ is such that, for all $\lambda$,
$\varepsilon$ and $\delta$ positive sufficiently small, one has
\begin{equation*}
  \mathbb{P}(\{|\omega_0-\zeta(\lambda)|\leq\varepsilon\})
  \geq e^{-\varepsilon^{-\delta}},
\end{equation*}
then, by virtue of~(\ref{eq:10}), for all $\lambda$, $\varepsilon$ and
$\delta$ positive sufficiently small, one has
\begin{equation*}
  \mathbb{P}(\{\omega^\pm_0 \geq\varepsilon\})\geq
  e^{-\varepsilon^{-\delta}}.
\end{equation*}
It is well known that, under this assumption, the Lifshitz exponent
for the density of states of the discrete Anderson model is equal to
$d/2$ (see e.g.~\cite{MR94h:47068}).\\
This completes the proof of Theorem~\ref{t2}.\qed
\subsection{The Wegner estimate}
\label{sec:wegner-estimate}
We now prove Theorem \ref{t3} using the results of~\cite{MR1934351}.
Let $H_{\lambda,r,\sigma,n}^P$ be the operator
$H_{\lambda,\omega,n}^P$ where the random variables
$(\omega_\gamma)_{\gamma\in\Z^d}$ are written in polar coordinates
i.e. $(\omega_\gamma)_{\gamma\in\Z^d}=
(r_\gamma(\omega)\,\sigma_\gamma(\omega))_{\gamma_\gamma\in\Z^d}$
where $r=(r_\gamma(\omega))_{\gamma\in\Z^d}$ has only non negative
components and $\sigma=(\sigma_\gamma(\omega))_{\gamma\in\Z^d}
\in[{\mathbb S}^{d-1}]^{\Z^d}$. Then, the basic observation is that
\begin{equation}
  \label{eq:19}
  \mathbb{P}(\text{dist}(\sigma(H_{\lambda,\omega,n}^P),E)\leq\varepsilon)
  = \mathbb{E}_\sigma\left(\mathbb{P}_r(\text{dist}(
    \sigma(H_{\lambda,r,\sigma,n}^P),E)\leq\varepsilon|\ \sigma)\right)
\end{equation}
where $\mathbb{P}_r(\cdot|\ \sigma)$ denotes the probability in the
$r$-variable conditioned on $\sigma$, and $\mathbb{E}_\sigma$, the
expectation in the $\sigma$-variable.\\
Now, fix $\sigma\in[{\mathbb S}^{d-1}]^{\Z^d}$. Using the notations of
section~\ref{sec:wegner-estimate-1}, we write
\begin{equation}
  \label{eq:21}
  H_{\lambda,\omega}=H_0+\lambda\sum_{\gamma\in\Z^d}
  r_\sigma(\omega_\gamma)v_{\sigma_\gamma}(\cdot-\gamma)
  +\lambda^2V_{2,\omega,\lambda}
\end{equation}
where
\begin{itemize}
\item $v_{\sigma_\gamma}=-\sigma_\gamma\cdot\nabla q$,
\item $V_{2,\omega,\lambda}$ is a potential bounded uniformly in
  $\lambda$ and $\omega$.
\end{itemize}
As $q$ is $C^2$ with compact support, for any $\sigma_0\in{\mathbb
  S}^{d-1}$, $v_{\sigma_0}$ is $C^1$ with compact support and does not
vanish identically. Assumptions (H.0.2) and (H.3) guarantee that the
random variables $(r_\gamma(\omega))_{\gamma\in\Z^d}$ are independent
and nicely distributed.\\
Hence, the model~\eqref{eq:21} satisfies the assumptions considered in
section 6 of~\cite{MR1934351} except for the fact that, in the present
case, $V_{2,\omega,\lambda}$ depends on $\lambda$. This does not
matter as it is bounded uniformly in $\lambda$. In particular, Lemma
6.1 of~\cite{MR1934351} from asserts that there exists $\lambda_0>0$
such that, for $\lambda\in(0,\lambda_0]$, there exists $C_\lambda>0$
such that, for all $E\in[E_{\lambda}, E_{\lambda}+\lambda/C]$ and
$\varepsilon >0$ such that
\begin{equation}
  \label{eq:22}
  \mathbb{P}(\text{dist}(\sigma(H_{\lambda,\omega,\sigma,n}^P),E)\leq
  \varepsilon|\ \sigma)\leq C_\lambda\left[\sup_{\gamma\in\Z^d}
    \|h_{\sigma_\gamma}'\|_\infty\right]\,\varepsilon^{\nu}n^d.
\end{equation}
The explicit form of the the constant appearing on the right side of
formula~\eqref{eq:22} is obtained by following the proof of Lemma 6.1
in~\cite{MR1934351}.\\
The bound~\eqref{eq:23} then guarantees that the sup in~\eqref{eq:22}
is essentially bounded as a function of $\sigma$. We then complete the
proof of Theorem~\ref{t3} by integrating~\eqref{eq:22} with respect to
$\sigma$ and using~\eqref{eq:19}.\qed
\section{Proof of Theorem~\ref{thr:3}}
\label{sec:proof-theor-refthr:3}
We now turn to the proof of Theorem~\ref{thr:3}. The proof follows the
spirit of~\cite{Klopp:1999,Ghribi:2006}.\\
For $\gamma\in\Z^d$, define
$\tilde{\omega}=(\tilde{\omega}_\gamma)_{\gamma \in
  {\zed}^d}=(\omega_\gamma-\zeta)_{\gamma \in {\zed}^d}$. Write
\begin{equation}
  \label{eq:25}
  V_{\lambda,\omega}=V_{\lambda,\overline{\zeta}}
  +\lambda\,\delta V_{\lambda,\tilde{\omega}}=
  V_{\lambda,\overline{\zeta}}+\lambda
  V_{1,\lambda,\tilde{\omega}} 
  +\lambda^2V_{2,\lambda,\tilde{\omega}}
\end{equation}
where
\begin{equation}
  \label{eq:24}
  V_{1,\lambda,\tilde{\omega}}=-\sum_{\gamma\in\Z^d}\nabla 
  q(x-\gamma-\lambda\zeta)\cdot\tilde\omega_\gamma.
\end{equation}
We decompose our random Hamiltonian $H_{\lambda,\tilde{\omega}}
:=H_{\lambda,\omega}$ on the transla\-tion-invariant subspaces
${\mathcal E}_{\lambda,\zeta,0}$ and ${\mathcal E}_{\lambda,\zeta,+}$
defined in the section~\ref{sec:reduction-procedure}. Thus, we obtain
the random operators
\begin{equation*}
  H_{\lambda,\tilde{\omega},0}=\Pi_{\lambda,\zeta,0}
  H_{\lambda,\tilde{\omega}}\Pi_{\lambda,\zeta,0}\quad\text{and}
  \quad  H_{\lambda,\tilde{\omega},+}=\Pi_{\lambda,\zeta,+} 
  H_{\lambda,\tilde{\omega}}\Pi_{\lambda,\zeta,+}.
\end{equation*}
In the orthogonal decomposition of $L^2({\er}^d)={\mathcal
  E}_{\lambda,\zeta,0}\overset{\perp}{\oplus} {\mathcal
  E}_{\lambda,\zeta,+}$, $H_{\lambda,\tilde{\omega}}$ is represented
by the matrix
\begin{equation}
  \label{eq:29}
  \begin{pmatrix}
    H_{\lambda,\tilde{\omega},0}&
    \lambda\,\Pi_{\lambda,\zeta,0}\,\delta
    V_{\lambda,\tilde{\omega}}\, \Pi_{\lambda,\zeta,+} \\
    \lambda \Pi_{\lambda,\zeta,+}\,\delta
    V_{\lambda,\tilde{\omega}}\,\Pi_{\lambda,\zeta,0}&
    H_{\lambda,\tilde{\omega},+}.
  \end{pmatrix}
\end{equation}
In section~\ref{sec:study-pi_lambda-0-5}, we give lower and upper
bounds on $H_{\lambda,\tilde{\omega},0}$ which we prove in
section~\ref{sec:proofs-prop-refpr1}. Theorem~\ref{thr:3} then follows
from the fact that the off-diagonal terms in~\eqref{eq:29} are
controlled by the diagonal ones; this is explained in
section~\ref{sec:lower-upper-bounds}.
\input preuve-red.tex
%

% \bibliographystyle{plain} 
% \bibliography{biblio}

\def\cprime{$'$} \def\cydot{\leavevmode\raise.4ex\hbox{.}}

\end{document}

%% file: preuve-red.tex
\subsection{The operator $H_{\lambda,\tilde{\omega},0}$}
\label{sec:study-pi_lambda-0-5}
In this section, using the non-degeneracy for the density of states of
$H_{\lambda,\overline{\zeta}}$ at $E(\lambda,\zeta)$, we give lower
and upper bounds on $H_{\lambda,\tilde{\omega},0}$.\\
As seen in section~\ref{sec:reduction-procedure}, the operator
$H_{\lambda,\tilde{\omega},0}$ is unitarily equivalent to the operator
$h_{\lambda,\tilde{\omega}}$ acting on $L^2(\T^*)$ and defined by
\begin{equation*}
  h_{\lambda,\tilde{\omega}}=h_{\lambda}+\lambda
  v_{1,\lambda,\tilde{\omega}} +
  {\lambda}^2 v_{2,\lambda,\tilde{\omega}}, 
\end{equation*}
where
\begin{itemize}
\item $h_{\lambda}$ is the multiplication by
  $E_0(\lambda,\zeta,\theta)$,
\item the operator $v_{1,\lambda,\tilde{\omega}}$ has the kernel
  \begin{equation}
    \label{eq:26}
    v_{1,\lambda,\tilde{\omega}}(\theta,
    \theta')=\langle V_{1,\lambda,\tilde{\omega}}\,
    \varphi_0(\lambda,\zeta,\theta,
    \cdot),\varphi_0(\lambda,\zeta,\theta',\cdot)\rangle_{L^2(K_0)},
  \end{equation}
\item the operator $v_{2,\lambda,\tilde{\omega}}$ has the kernel
  \begin{equation*}
    v_{2,\lambda,\tilde{\omega}}(\theta, \theta')=\langle
    V_{2,\lambda,\tilde{\omega}}\,\varphi_0(\lambda,\zeta,\theta,
    \cdot), \varphi_0(\lambda,\zeta,\theta', \cdot)\rangle_{L^2(K_0)} .
  \end{equation*}
\end{itemize}
The potential $V_{1,\lambda,\tilde{\omega}}$ and
$V_{2,\lambda,\tilde{\omega}}$ are defined in~(\ref{eq:25})
and~(\ref{eq:24}). They are bounded uniformly in all parameters. This
will be used freely without a special mention. \\
We now recall a number of facts and definitions taken
from~\cite{Klopp:1999}. Let $t \in L^2(\T^*, {\mathcal
  H}_{\theta})$. We define the operator $P_t:\ L^2(\T^*)
\rightarrow L^2({\er}^d)$ by
\begin{equation*}
  \forall u \in L^2(\T^*),\ [P_t(u)](x)=\int_{\T^*}
  t(\theta,x) u(\theta) d\theta.
\end{equation*}
It satisfies 
\begin{equation}
  \label{eq:32}
  \|P_t\|_{L^2(\T^*)
    \rightarrow L^2({\er}^d)}\leq \|t\|_{L^2(\T^*, {\mathcal
      H}_{\theta})}.
\end{equation}
As the Floquet eigenvalue $E_0(\lambda,\zeta,\theta)$ is simple in a
neighborhood of $0$, the Floquet eigenvector
$\varphi_0(\lambda,\zeta,\theta, \cdot)$ is analytic in this
neighborhood.\\
Recall that $\varpi$ is defined in~(\ref{eq:27}). We
define the functions $\varphi_{0,\lambda,\zeta},
\,\tilde{\varphi}_{0,\lambda,\zeta}$ and $\delta
\varphi_{0,\lambda,\zeta}$ in $L^2(\T^*, {\mathcal H}_{\theta})$ by
\begin{gather*}
  \varphi_{0,\lambda,\zeta}(\theta,x)=\varphi_0
  (\lambda,\zeta,\theta;x),\quad
  \tilde{\varphi}_{0,\lambda,\zeta}(\theta,x)=
  \varphi_{0,\lambda,\zeta}(0,x) e^{i \theta \cdot x}\\
  \delta \varphi_{0,\lambda,\zeta}(\theta,x)=
  \frac{1}{\sqrt{\varpi(\theta)}}(\varphi_0(\lambda,.\zeta,\theta;x)
  -\tilde{\varphi}_{0,\lambda,\zeta}(\theta,x)).
\end{gather*}
Furthermore, these functions are bounded in $L^2(\T^*,
{\mathcal H}_{\theta})$ uniformly in $\zeta$ and $\lambda$ small.\\
Finally, we note that, for $u \in L^2(\T^*)$,
\begin{equation}
  \label{f1}
  P_{\varphi_{0,\lambda,\zeta}}(u)=
  P_{\tilde{\varphi}_{0,\lambda,\zeta}}(u)+P_{\delta
    \varphi_{0,\lambda,\zeta}}(\sqrt{\varpi} u). 
\end{equation} 
\begin{Rem}
  \label{k1}
  It is proved in~\cite{Klopp:1999} that, there exits $C>1$ such that,
  as operators on $L^2(\T)$, one has
  \begin{equation*}
    \frac{1}{C}\, \varpi\, \leq h_{\lambda}-E(\lambda,\zeta)\leq
    C\, \varpi. 
  \end{equation*}
\end{Rem}
\subsubsection{Lower and upper bounds on
  $v_{1,\lambda,\tilde{\omega}}$ and $v_{2,\lambda,\tilde{\omega}}$}
\label{sec:lower-bound-v_1}
\begin{Pro}
  \label{pr1}
  Recall that $v(\lambda,\zeta)$ is defined in~(\ref{eq:9}). There
  exists $C>0$ such that, for $u \in L^2(\T^*)$ and $\alpha>0$, we
  have
  \begin{multline}
    \label{eq:30}
    \left|\langle v_{1,\lambda,\tilde{\omega}}u,u\rangle -
      \sum_{\gamma \in {\zed}^d} [v(\lambda,\zeta)
      \,\cdot \tilde{\omega}_{\gamma}] \cdot |\hat{u}(\gamma)|^2\right|\\
    \leq C \left(\alpha \, \sum_{\gamma \in\Z^d} \|
      \tilde{\omega}_{\gamma} \|^2 \cdot |\hat{u}(\gamma)|^2
      +\left(1+\frac{1}{\alpha}\right) \langle \varpi u,u\rangle \right)
  \end{multline}
  and
  \begin{equation}
    \label{eq:31}
    \begin{split}
    |\langle v_{2,\lambda,\tilde{\omega}}u,u\rangle|&+
    \|V_{1,\lambda,\tilde{\omega}} P_{\varphi_{0,\lambda,\zeta}}(u)
    \|^2 + \|V_{2,\lambda,\tilde{\omega}}
    P_{\varphi_{0,\lambda,\zeta}}(u) \|^2 \\&\leq C \left(\sum_{\gamma
        \in {\zed}^d} \|\tilde{\omega}_{\gamma}\|^2
      \cdot |\hat{u}(\gamma)|^2 +\langle \varpi u,u\rangle\right).      
    \end{split}
  \end{equation}
\end{Pro}
\noindent Proposition~\ref{pr1} is proved in
section~\ref{sec:proofs-prop-refpr1}. We now use these results to give
lower and upper bounds on $H_{\lambda,\tilde{\omega},0}$.
\subsubsection{Lower and upper bounds for
  $H_{\lambda,\tilde{\omega},0}-E(\lambda,\zeta)$}
\label{sec:lower-bound-h}
We prove
\begin{Pro}
  \label{pr3}
  Under assumptions (H.0) and (H.1), there exists $C_0>0$ such that,
  for $\alpha>0$, there exists $\lambda_\alpha>0$ and $C_\alpha>0$
  such that, for all $\lambda \in [0,\lambda_0]$, on ${\mathcal
    E}_{\lambda,\zeta,0}$, one has
  \begin{equation*}
    \frac{1}{C_0} P_{\lambda,\zeta}
    h^-_{\lambda,\tilde{\omega},\zeta}P_{\lambda,\zeta}^* \leq
    \tilde{H}_{\lambda,\tilde{\omega},0}:=H_{\lambda,\tilde{\omega},0}
    -E(\lambda,\zeta) \leq C_0
    P_{\lambda,\zeta}h^+_{\lambda,\tilde{\omega},\zeta}P^*_{\lambda,\zeta},
  \end{equation*}
  where $h^{\pm}_{\lambda,\tilde{\omega},\zeta}$ are the random
  operators defined in Theorem~\ref{thr:3}.
\end{Pro}
\begin{proof}[Proof of Proposition~\ref{pr3}.]
  For $\lambda$ small, Proposition~\ref{pr1} and Remark~\ref{k1} imply
  that, there exists $C_0>0$ such that, for $\alpha>0$, there exists
  $\lambda_\alpha>0$ such that for all $\lambda
  \in[0,\lambda_\alpha]$,
  \begin{equation*}
    \frac1{C_0}\varpi +\lambda \sum_{\gamma \in {\zed}^d}
    [v(\lambda,\zeta) \,\cdot \tilde{\omega}_{\gamma}- C_0\alpha \|
    \tilde{\omega}_{\gamma} \|^2] \cdot \Pi_{\gamma} \leq
    h_{\lambda,\tilde{\omega}} -E(\lambda,\zeta)
  \end{equation*}
  and
  \begin{equation*}
    h_{\lambda,\tilde{\omega}}-E(\lambda,\zeta) \leq C_0\,\varpi + \lambda
    \sum_{\gamma \in {\zed}^d} [v(\lambda,\zeta)
    \,\cdot \tilde{\omega}_{\gamma}+C_0\alpha \| \tilde{\omega}_{\gamma}
    \|^2] \cdot \Pi_{\gamma}.
  \end{equation*}
  As $H_{\lambda,\tilde{\omega},0}$ and $h_{\lambda,\tilde{\omega}}$
  are unitarily equivalent, this completes the proof of
  Proposition~\ref{pr3}.
\end{proof}
\subsection{The operator $H_{\lambda,\tilde{\omega},+}$}
\label{sec:lower-upper-bounds}
By the definition of $\Pi_{\lambda,\zeta,+}$, there exists $\eta>0$
such that, for $\lambda$ sufficiently small,
\begin{equation*}
   (E(\lambda,\zeta)+\eta)\Pi_{\lambda,\zeta,+} \leq
  \Pi_{\lambda,\zeta,+}H_{\lambda,\overline{\zeta}}\Pi_{\lambda,\zeta,+}.
\end{equation*}
Let $\tilde{H}_{\lambda,\overline{\zeta},+}=
H_{\lambda,\overline{\zeta},+}-E(\lambda,\zeta)$ and
$\tilde{H}_{\lambda,\tilde\omega,+}=
H_{\lambda,\tilde\omega,+}-E(\lambda,\zeta)$ (see~\eqref{eq:29}). As
$|V_{1,\lambda,\tilde{\omega}}|$ and $|V_{1,\lambda,\tilde{\omega}}|$
are bounded, for $\lambda$ sufficiently small, one has
\begin{equation}
  \label{v1} 
  \frac{\eta}{2}\,\Pi_{\lambda,\zeta,+} \leq\frac{1}{2}
  \tilde{H}_{\lambda,\overline{\zeta},+} \leq
  \tilde{H}_{\lambda,\tilde{\omega},+}
  \leq 2\tilde{H}_{\lambda,\overline{\zeta},+}.
\end{equation}
\subsection{The proof of Theorem~\ref{thr:3}}
\label{sec:pf_th21}
For $\varphi=\varphi_0+\varphi_+\in {\mathcal
  E}_{\lambda,\zeta,0}\overset{\perp}{\oplus}{\mathcal
  E}_{\lambda,\zeta,+}$, by~(\ref{eq:29}), one has
\begin{equation*}
  \begin{split}
    &\left| \langle \tilde{H}_{\lambda,\tilde{\omega}} \varphi,
      \varphi\rangle- \langle \tilde{H}_{\lambda,\tilde{\omega},0}
      \varphi_0, \varphi_0\rangle-\langle
      \tilde{H}_{\lambda,\tilde{\omega},+}
      \varphi_+,\varphi_+\rangle\right|\\ &\hskip2cm\leq 2 \lambda
    |\langle V_{1,\lambda,\tilde{\omega}} \varphi_+,
    \varphi_0\rangle| +2 {\lambda}^2 |\langle
    V_{2,\lambda,\tilde{\omega}}\varphi_+, \varphi_0\rangle|.
  \end{split}
\end{equation*}
Using the Cauchy-Schwarz inequality, we get
\begin{equation*}
  |\langle V_{1,\lambda,\tilde{\omega}}\varphi_+,\varphi_0\rangle|+
  |\langle V_{2,\lambda,\tilde{\omega}} \varphi_+, \varphi_0\rangle|
  \leq 2 \|V_{1,\lambda,\tilde{\omega}}\varphi_0 \|^2+
  2\|V_{2,\lambda,\tilde{\omega}}\varphi_0 \|^2+4\|\varphi_+\|^2.
\end{equation*}
Then, the decomposition~\eqref{eq:29} and~\eqref{v1} give
\begin{multline}
  \label{f}
  \frac{1}{C} \begin{pmatrix}
    \tilde{H}_{\lambda,\tilde{\omega},0}-C\lambda K_0
    &0\\
    0&\tilde{H}_{\lambda,\tilde{\omega},+}
    -C\lambda\Pi_{\lambda,\zeta,+}
  \end{pmatrix}
  \leq \tilde{H}_{\lambda,\tilde{\omega}}\\=
  \tilde{H}_{\lambda,\tilde{\omega}}\leq C\begin{pmatrix}
    \tilde{H}_{\lambda,\tilde{\omega},0}+C\lambda K_0&0\\
    0&\tilde{H}_{\lambda,\tilde{\omega},+}+
    C\lambda\Pi_{\lambda,\zeta,+} 
  \end{pmatrix}
\end{multline}
where
\begin{equation*}
  K_0=
  \Pi_{\lambda,\zeta,0}(V_{1,\lambda,\tilde{\omega}}^2+
  V_{2,\lambda,\tilde{\omega}}^2)\Pi_{\lambda,\zeta,0}.
\end{equation*}
The estimate~(\ref{s1}) of Proposition~\ref{pr1} implies that
\begin{equation*}
  K_0\leq C P_{\lambda,\zeta}\left(\varpi+\sum_{\gamma\in{\zed}^d}
    \|\tilde{\omega}_{\gamma}\|^2\Pi_\gamma\right)
  P^*_{\lambda,\zeta}.
\end{equation*}
On the other hand,~(\ref{v1}) implies that, for $\lambda$
sufficiently small
\begin{equation*}
  \frac12\tilde{H}_{\lambda,\tilde{\omega},+} \leq
  \tilde{H}_{\lambda,\tilde{\omega},+}-C\lambda
  \Pi_{\lambda,\zeta,+}   \leq
  \tilde{H}_{\lambda,\tilde{\omega},+}+C\lambda
  \Pi_{\lambda,\zeta,+}\leq
  2\tilde{H}_{\lambda,\tilde{\omega},+}.
\end{equation*}
Combining these two estimates with~(\ref{f}) and
Proposition~\ref{pr3}, we complete the proof of
Theorem~\ref{thr:3}.
\subsection{The proof of Propositions~\ref{pr1}}
\label{sec:proofs-prop-refpr1}
We first prove
\begin{Le}
  \label{r}
  There exists a constant $C>0$ such that, for all $u \in L^2({\mathbb
    T}^d)$ and $\alpha>0$, one has
  \begin{gather}
    \label{s}
    \begin{split}
      &\left|\langle V_{1,\lambda,\tilde{\omega}}
        P_{\tilde{\varphi}_{0,\lambda,\zeta}}
        (u),P_{\tilde{\varphi}_{0,\lambda,\zeta} }(u)\rangle
        -\sum_{\gamma \in {\zed}^d}
        [v(\lambda,\zeta)\,\cdot \tilde{\omega}_{\gamma}]\;
        |\hat{u}(\gamma)|^2\right|\\
      &\hskip1.5cm\leq C\left(\alpha\,\sum_{\gamma \in {\zed}^d} \|
        \tilde{\omega}_{\gamma} \|^2 \cdot |\hat{u}(\gamma)|^2+
        \left(1+\frac{4}{\alpha}\right) \langle \varpi
        u,u\rangle\right),
    \end{split}\\
    \label{s1}
    \begin{split}
      &\|V_{1,\lambda,\tilde{\omega}}
      P_{\tilde{\varphi}_{0,\lambda,\zeta}}(u) \|^2 +
      \|V_{2,\lambda,\tilde{\omega}}
      P_{\tilde{\varphi}_{0,\lambda,\zeta}}(u) \|^2\\ &\hskip2cm\leq
      C\left(\alpha\sum_{\gamma\in{\zed}^d} \|\tilde{\omega}_{\gamma}
        \|^2\cdot |\hat{u}(\gamma)|^2 +\left(1+\frac{4}{\alpha}\right)
        \langle \varpi u,u\rangle\right),
    \end{split}
  \end{gather}
\end{Le}
\begin{proof}[Proof of Lemma~\ref{r}.]
  We compute
  \begin{multline*}
    \langle V_{1,\lambda,\tilde{\omega}}
    P_{\tilde{\varphi}_{0,\lambda,\zeta}}(u),
    P_{\tilde{\varphi}_{0,\lambda,\zeta}}(u)\rangle \\=-\sum_{\gamma
      \in {\zed}^d} \tilde{\omega}_{\gamma} \cdot \int_{{\er}^d} \nabla
    q(x-\lambda \zeta) |\varphi_0(\lambda,\zeta,0;x)|^2
    \cdot|\phi_{\gamma}(u)(x)|^2 dx,
  \end{multline*}
  where $\phi_{\gamma}(u)(x)=\int_\T
  e^{i \theta \cdot \gamma} \cdot e^{i \theta \cdot x} u(\theta) d\theta.$ \\
  Recall that $0$ is the unique zero of $\varpi$ on $\T^*$ and it is
  non-degenerate. Thus, the function
  $g(\theta,x)=\varpi(\theta)^{-1/2}(e^{i\theta\cdot x} -1)$ is defined on
  $\T\times {\er}^d$ and
  \begin{equation*}
    \sup_{(\theta,x)\in \T^*
      \times \er}(1 +|x|)^{-1}|g(\theta,x)|<+\infty .
  \end{equation*}
  For $\gamma\in{\zed}^d$, $u\in L^2(\T^*)$ and $x\in\er^d$, one has
  \begin{equation*}
    \psi_{\gamma}(u)(x)=\phi_{\gamma}(u)(x)-\hat{u}(\gamma)
    =\int_{\T^*} g(\theta,x)e^{i
      \gamma \cdot \theta} \sqrt{\varpi(\theta)} u(\theta)
    d\theta.
  \end{equation*}
  Note that
  \begin{equation}
    \label{eq:2}
    \begin{split}
      \sum_{\gamma\in\Z^d}|\psi_{\gamma}(u)(x)|^2&=\int_{\T^*}
      |g(\theta,x)|^2 |\sqrt{\varpi(\theta)}u(\theta)|^2
      d\theta\\&\leq C(1+|x|)^2\langle\varpi u,u\rangle.
    \end{split}
  \end{equation}
  Recall that $v(\lambda,\zeta)$ is defined by~\eqref{eq:9}. We define
  \begin{gather*}
    v'_{1,\lambda,\tilde{\omega}}[u]=-\sum_{\gamma\in{\zed}^d}
    \tilde{\omega}_{\gamma} \cdot \int_{{\er}^d} \nabla q(x-\lambda \zeta)
    |\varphi_0(\lambda,\zeta,0;x)|^2 |\psi_{\gamma}(u)(x)|^2dx,\\
    v''_{1,\lambda,\tilde{\omega}}[u] =\langle
    V_{1,\lambda,\tilde{\omega}}
    P_{\tilde{\varphi}_{0,\lambda,\zeta}}(u),
    P_{\tilde{\varphi}_{0,\lambda,\zeta}}(u)\rangle- \sum_{\gamma \in
      {\zed}^d} v(\lambda,\zeta) \cdot \tilde{\omega}_{\gamma}
    |\hat{u}(\gamma)|^2- v'_{1,\lambda,\tilde{\omega}}[u].
  \end{gather*}
  As the random variables $(\tilde{\omega}_{\gamma})_{\gamma \in
    {\zed}^d}$ are bounded, by~\eqref{eq:2}, we compute
  \begin{equation*}
    \begin{split}
      |v'_{1,\lambda,\tilde{\omega}}[u]| &\leq C\int_{{\er}^d}
      \|\nabla q(x-\lambda \zeta)\| |\varphi_0(\lambda,\zeta,0;x)|^2
      \sum_{\gamma\in{\zed}^d} |\psi_{\gamma}(u)(x)|^2dx \\ &\leq
      C\langle\varpi\,u,u\rangle\int_{{\er}^d} \|\nabla q(x-\lambda
      \zeta)\| |\varphi_0(\lambda,\zeta,0;x)|^2 (1+|x|)^{2}dx \\ &\leq
      C\langle\varpi\,u,u\rangle.
    \end{split}
  \end{equation*}
  By the Cauchy-Schwarz inequality, one has
  \begin{equation*}
    \begin{split}
      |v''_{1,\lambda,\tilde{\omega}}[u]|&=2\left|\text{Re}
        \left(\sum_{\gamma\in{\zed}^d} \hat{u}(\gamma)
          \tilde{\omega}_{\gamma}\cdot \int_{{\er}^d} \nabla q(x-\lambda
          \zeta) |\varphi_0(\lambda,\zeta,0;x)|^2
          \overline{\psi_{\gamma}(u)(x) }dx\right)\right| \\&\leq
      \alpha\left[\int_{{\er}^d} \|\nabla q(x-\lambda \zeta)\|\,
        |\varphi_0(\lambda,\zeta,0;x)|^2 dx\right] \sum_{\gamma \in
        {\zed}^d} \|\tilde{\omega}_{\gamma}\|^2\,|\hat{u}(\gamma)|^2
      \\&\hskip1cm+ \frac{4}{\alpha} \sum_{\gamma \in {\zed}^d}
      \int_{{\er}^d} \|\nabla q(x-\lambda \zeta)\|
      \,|\varphi_0(\lambda,\zeta,0;x)|^2 | \psi_{\gamma}(u)(x)|^2 dx.
    \end{split}
  \end{equation*}
  Using~\eqref{eq:2}, we obtain that
  \begin{equation*}
    |v''_{1,\lambda,\tilde{\omega}}[u]|\leq  C\left( \alpha
      \sum_{\gamma \in
        {\zed}^d} \|\tilde{\omega}_{\gamma}\|^2 \cdot |\hat{u}(\gamma)|^2 +
      \frac4{\alpha}\langle \varpi u,u \rangle\right).
  \end{equation*}
  Finally, adding this to the estimate for
  $|v'_{1,\lambda,\tilde{\omega}}[u]|$, we get
  \begin{multline*}
    \left|\langle V_{1,\lambda,\tilde{\omega}}
      P_{\tilde{\varphi}_{0,\lambda,\zeta}}(u),
      P_{\tilde{\varphi}_{0,\lambda,\zeta}}(u)\rangle-\sum_{\gamma \in
        {\zed}^d} v(\lambda,\zeta)\,\cdot \tilde{\omega}_{\gamma}
      |\hat{u}(\gamma)|^2\right| \\\leq C\left( \alpha \sum_{\gamma
        \in {\zed}^d} \|\tilde{\omega}_{\gamma}\|^2
      \cdot |\hat{u}(\gamma)|^2+\left(1+\frac{4}{\alpha}\right) \langle
      \varpi u,u \rangle\right).
  \end{multline*}
  This completes the proof of~\eqref{s}.\\
  The two terms in the left hand side of~(\ref{s1}) are dealt with in
  the same way; so, we only give the details for
  $\|V_{1,\lambda,\tilde{\omega}}
  P_{\tilde{\varphi}_{0,\lambda,\zeta}}(u)\|^2$. We compute
  \begin{equation*}
    \begin{split}
      &\|V_{1,\lambda,\tilde{\omega}}
      P_{\tilde{\varphi}_{0,\lambda,\zeta} }(u)\|^2 = \int_{{\er}^d}
      \left|V_{1,\lambda,\tilde{\omega}}(x)\,
        \varphi_0(\lambda,\zeta,0;x)\,\phi(u)(x)\right|^2 dx
      \\&\hskip1cm= \int_{{\er}^d} \left|\sum_{\gamma\in\Z^d}\nabla
        q(x-\lambda\zeta-\gamma)\cdot\tilde\omega_\gamma\right|^2\,
      |\varphi_0(\lambda,\zeta,0;x)\,\phi(u)(x)|^2 dx\\& \hskip1cm\leq
      C\sum_{\gamma\in\Z^d} \|\tilde\omega_\gamma\|^2
      \int_{{\er}^d}\|\nabla q(x-\lambda\zeta)\|^2\,
      |\varphi_0(\lambda,\zeta,0;x)\,\phi_\gamma(u)(x)|^2 dx
    \end{split}
  \end{equation*}
  where, in the last step, as $q$ is compactly supported, the number
  of non vanishing terms of the sum inside the integral is bounded
  uniformly. \\
  Now, by the definition of $\phi_\gamma$ and $\psi_\gamma$, we have
  \begin{equation*}
    \begin{split}
      &\int_{{\er}^d} |\nabla q(x-\lambda\zeta)|^2
      |\varphi_0(\lambda,\zeta,0;x)|^2
      |\phi_{\gamma}(u)(x)|^2dx\\&\hskip1cm\leq 2\int_{{\er}^d}
      |\nabla q(x-\lambda\zeta)|^2
      |\varphi_0(\lambda,\zeta,0;x)|^2\left(|\hat
        u(\gamma)|^2+|\psi_{\gamma}(u)(x)|^2\right)dx \\&\hskip1cm
      \leq C|\hat u(\gamma)|^2+C \int_{{\er}^d} |\nabla
      q(x-\lambda\zeta)|^2 |\varphi_0(\lambda,\zeta,0;x)|^2
      |\psi_{\gamma}(u)(x)|^2 dx.
    \end{split}
  \end{equation*}
  We plug this into the estimate for $\|V_{1,\lambda,\tilde{\omega}}
  P_{\tilde{\varphi}_{0,\lambda,\zeta} }(u)\|^2$ and~\eqref{eq:2}
  yields
  \begin{equation*}
    \begin{split}
      &\|V_{1,\lambda,\tilde{\omega}}
      P_{\tilde{\varphi}_{0,\lambda,\zeta} }(u)\|^2 \\&\hskip1cm\leq
      C\int_{{\er}^d}|\nabla q(x-\lambda\zeta)|^2
      |\varphi_0(\lambda,\zeta,0;x)|^2
      \cdot\sum_{\gamma\in\Z^d}|\phi_{\gamma}(u)(x)|^2dx \\&\hskip6cm+
      C\sum_{\gamma\in\Z^d} \|\tilde\omega_{\gamma}\|^2|\hat
      u(\gamma)|^2\\&\hskip1cm\leq C\|\tilde\omega_{\gamma}\|^2 |\hat
      u(\gamma)|^2+ C\langle \varpi u,u \rangle.
    \end{split}
  \end{equation*}
  The computation for $\|V_{2,\lambda,\tilde{\omega}}
  P_{\tilde{\varphi}_{0,\lambda,\zeta} }(u)\|^2$ is the same as
  \begin{equation*}
    V_{2,\lambda,\tilde{\omega}}=\sum_{\gamma\in\Z^d}
    q(x-\gamma,\zeta,\tilde\omega_\gamma)
  \end{equation*}
  where, denoting the Hessian of $q$ at $x$ by $Q$, we have
  \begin{equation*}
    q(x,\zeta,\tilde\omega_\gamma)=
    \int_0^1\langle Q(x-\lambda(\zeta+t\tilde\omega_\gamma))
    \tilde\omega_\gamma,\tilde\omega_\gamma\rangle(1-t)dt.
  \end{equation*}
  So~\eqref{s1} is proved and the proof of Lemma~\ref{r} is complete.
\end{proof}
\begin{proof}[Proof of Proposition~\ref{pr1}]
  Using~(\ref{eq:26}) and~\eqref{f1}, we write
  \begin{equation*}
    \begin{split}
      \langle v_{1,\lambda,\tilde{\omega},\zeta} u,u \rangle&= \langle
      V_{1,\lambda,\tilde{\omega}} P_{\delta
        \varphi_{0,\lambda,\zeta}}(\sqrt{\varpi} u),P_{\delta
        \varphi_{0,\lambda,\zeta}}(\sqrt{\varpi} u)\rangle
      \\
      &\hskip1cm+\langle V_{1,\lambda,\tilde{\omega}}
      P_{\tilde{\varphi}_{0,\lambda,\zeta}}(u),
      P_{\tilde{\varphi}_{0,\lambda,\zeta}}(u)\rangle\\
      &\hskip2cm+ 2 \text{Re}(\langle V_{1,\lambda,\tilde{\omega}}
      P_{\tilde{\varphi}_{0,\lambda,\zeta}}(u),P_{\delta
        \varphi_{0,\lambda,\zeta}}(\sqrt{\varpi} u)\rangle ).
    \end{split}
  \end{equation*}
  Hence, as $V_{1,\lambda,\tilde{\omega}}$ is bounded,
  \begin{equation}
    \label{gf1}
    \begin{split}
      &\left|\langle v_{1,\lambda,\tilde{\omega}} u,u\rangle-
        \sum_{\gamma \in {\zed}^d} v(\lambda,\zeta)
        \cdot \tilde{\omega}_{\gamma} |\hat{u}(\gamma)|^2\right|\\
      &\leq\left|\langle
        V_{1,\lambda,\tilde{\omega}}P_{\tilde{\varphi}_{0,\lambda,\zeta}}(u),
        P_{\tilde{\varphi}_{0,\lambda,\zeta}}(u)\rangle-\sum_{\gamma
          \in {\zed}^d} v(\lambda,\zeta) \cdot \tilde{\omega}_{\gamma}
        |\hat{u}(\gamma)|^2\right| \\
      &\hskip.5cm+2 \left|\langle V_{1,\lambda,\tilde{\omega}}
        P_{\tilde{\varphi}_{0,\lambda,\zeta}}(u),P_{\delta
          \varphi_{0,\lambda,\zeta}}(\sqrt{\varpi} u)\rangle\right|
      +\left\|P_{\delta \varphi_{0,\lambda,\zeta}}(\sqrt{\varpi}
        u)\right\|^2.
    \end{split}
  \end{equation}
  The Cauchy-Schwarz inequality then yields
  \begin{multline*}
    |\langle V_{1,\lambda,\tilde{\omega}}
    P_{\tilde{\varphi}_{0,\lambda,\zeta}}(u),P_{\delta
      \varphi_{0,\lambda,\zeta}}(\sqrt{\varpi} u)\rangle|\\ \leq
    \alpha \|V_{1,\lambda,\tilde{\omega}}
    P_{\tilde{\varphi}_{0,\lambda,\zeta}}(u)\|^2 + \frac{4}{\alpha}
    \|P_{\delta \varphi_{0,\lambda,\zeta}}(\sqrt{\varpi} u)\|^2.
  \end{multline*}
  Combining this with~\eqref{gf1},~\eqref{s1} and~\eqref{s}, we obtain
  \begin{multline*}
    \left|\langle v_{1,\lambda,\tilde{\omega}} u,u\rangle-\sum_{\gamma
        \in {\zed}^d}[v(\lambda,\zeta) \cdot \tilde{\omega}_{\gamma}] \cdot 
      |\hat{u}(\gamma)|^2\right| \\\leq C\left( \alpha \sum_{\gamma
        \in {\zed}^d} \|\tilde{\omega}_{\gamma}\|^2
      \cdot |\hat{u}(\gamma)|^2 + \left(1+\frac{1}{\alpha}\right)\langle 
      \varpi u,u\rangle\right).
  \end{multline*}
  This completes the proof of~(\ref{eq:30}).\\
  Using~\eqref{f1} and the expansion done above for $\langle
  v_{1,\lambda,\tilde{\omega}} u,u\rangle$ , we compute
  \begin{equation*}
    |\langle v_{2,\lambda,\tilde{\omega}}u,u\rangle|
    \leq 2 \|V_{2,\lambda,\tilde{\omega}}
    P_{\tilde{\varphi}_{0,\lambda,\zeta}}(u)\|^2
    + 2 \|P_{{\delta\varphi}_{0,\lambda,\zeta}}(\sqrt{\varpi}u) \|^2.
  \end{equation*}
  Combining this with~\eqref{s1}, we get
  \begin{equation*}
    |\langle v_{2,\lambda,\tilde{\omega}}u,u\rangle| \leq C
    \left(\sum_{\gamma \in
        {\zed}^d} \|\tilde{\omega}_{\gamma} \|^2 \cdot |u(\gamma)|^2 +
      \langle \varpi u, u\rangle    \right).
  \end{equation*}
  The estimates for $\|V_{1,\lambda,\tilde{\omega}}
  P_{\varphi_{0,\lambda,\zeta}}(u) \|^2$ and
  $\|V_{2,\lambda,\tilde{\omega}} P_{\varphi_{0,\lambda,\zeta}}(u)
  \|^2$ are obtained in the same way. This completes the proof
  of~(\ref{eq:31}), hence, of Proposition~\ref{pr1}.
\end{proof}

%%% Local Variables: 
%%% mode: latex
%%% TeX-master: "displacement"
%%% End: 